\documentclass{amsart}

\usepackage[utf8]{inputenc}
\usepackage{amsmath,amsfonts,amsthm,xypic}
\xyoption{all}
\def \Hil {\mathcal{H}}

\def \R {\mathbb{R}}

\def \C {\mathbb{C}}

\newcommand{\pd}[1]{\frac{\partial }{\partial #1}}
\newcommand{\pdc}[2]{\frac{\partial #1}{\partial #2}}

\newtheorem{definition}{Definition}
\newtheorem{lemma}{Lemma}

\newtheorem{note}{Note}

\begin{document}
\title{The Ehrenfest picture and the geometry of Quantum Mechanics}

\author{J. Clemente-Gallardo}
\address{ BIFI-Departamento de F\'{\i}sica Te\'orica and
Unidad Asociada IQFR-BIFI,  Universidad de
  Zaragoza,  Edificio I+D-Campus   R\'{\i}o Ebro, Mariano Esquillor
  s/n, 50018 Zaragoza (Spain)}  
\email{jesus.clementegallardo@bifi.es}

\author{G. Marmo}
\address{
Dipartimento di Fisica dell'Universit\`a di
 Napoli ``Federico II'' and INFN--Sezione di Napoli, Complesso
  Universitario di Monte Sant'Angelo, via Cintia, I-80126 Napoli,
  (Italy)}
\email{marmo@na.infn.it}

\maketitle

\begin{abstract}
In this paper we develop a picture of Quantum Mechanics based on the
description of physical observables in terms of expectation value
functions, generalizing thus the so called Ehrenfest theorems for
quantum dynamics. Our basic technical ingredient is 
 the set of tools which has been developed in the last years for the geometrical
formulation of Quantum Mechanics. In the new picture, we analyze the problem of the
dynamical equations, the uncertainty relations and interference and
illustrate the construction with the simple case of a two-level system.
\end{abstract}

\section{Introduction}

The paradigmatic example of matter waves (see
\cite{ClementeGallardo:2008p614}), e.g., electron
interference,  shows very neatly that we have at least three important
aspects  of quantum systems:
\begin{itemize}
\item a wave-like behavior incorporated in the Schrödinger picture;
\item a corpuscular-like behavior at the detector  giving rise to the
  Heisenberg picture;
\item and a probabilistic-statistical behavior which emerges from the
  erratic behavior of the clicking of the detectors. 
\end{itemize}
A closer scrutiny of this last aspect suggests that a good approach to
the description of a quantum system would be to formulate Quantum
Mechanics in terms of expectation value functions.

Let us assume that we have multiple copies of a given quantum system
and let us consider the experiment in which we measure the position of
the detector which clicks.  After a sufficient long time, we would end
up with a set of values in the form
\begin{equation}
  \label{eq:1}
  e_{Q}(\psi)=\frac{\langle\psi|\hat Q|\psi\rangle}{\langle\psi|\psi\rangle} 
\end{equation}
where $|\psi\rangle$ represents in the Dirac bra-ket notation the
vectors of a Hilbert space ${\mathcal H}$ and $\hat Q$ represents the
position operator acting on ${\mathcal H}$. 
More generally, if we perform measurements of other observable, say
$\hat A$, we would obtain 
\begin{equation}
  \label{eq:2}
  e_{A}(\psi)=\frac{\langle\psi|\hat A|\psi\rangle}{\langle\psi|\psi\rangle}. 
\end{equation}
Thus the probabilistic-statistical aspect of Quantum Mechanics is
captured by elaborating a picture completely given in terms of
expectation value functions. Historically, expectation value functions
appear already in the so-called \textit{Ehrenfest theorem} (see
\cite{Ehrenfest1927}) . In this
work, we ellaborate on it to define an alternative picture of Quantum
Mechanics.

We will argue that expectation value functions are able to provide an
alternative picture of Quantum Mechanics with respect to the
Schrödinger or the Heisenberg ones.  For instance, for a system
described by a Hamiltonian operator in the form
$$
\hat H=\frac 1 {2m}\hat P ^{2}+\hat V(x),
$$
we would obtain a dynamical system at the level of the expectation
value functions as
\begin{equation}
  \label{eq:4}
  \frac d{dt} e_{Q}(\psi)=\frac 1 m e_{P}(\psi); \qquad
  \frac d{dt} e_{Q}(\psi)=-e_{\mathrm{grad} V}(\psi);
\end{equation}
which is usually known as the Ehrenfest theorem. It is appropriate to
remark that already Koopman \cite{Koopman1931} and von Neuman
\cite{Neumann1932} proved that both Classical
and Quantum Mechanics can be treated in this picture.
We recall that while Quantum Mechanics  would be formulated on the
Hilbert space of square-integrable complex-valued functions on the
``configuration'' or the ``momentum'' space to ensure the
irreducibility of the representation of the Heisenberg-Weyl algebra; Classical Mechanics would
be formulated on the Hilbert space of square-integrable complex-valued
functions on the full phase-space.

\section{The tensors of the geometric formulation of Quantum
  Mechanics} 
\label{sec:geom-form-quant}

The aim of this section is simply to provide a
tensorial characterization of Quantum 
Mechanics.  We shall see later how  the tensors obtained  now will
provide us with the necessary tools for our definition of the
Ehrenfest picture of Quantum Mechanics.  
The section contains a brief
summary of the results presented in several recent works as
\cite{Carinena2007b,Carinena2007a,ClementeGallardo:2008p614,Aniello,Aniello2011}. We
also address the interested reader to some other references by serveral
authors covering similar topics. Just
to mention the most relevant references ordered chronologically, let us refer to former
interesting approaches as \cite{STROCCHI1966}, the seminal
work by Kibble \cite{Kibble:1979p7279}, the works by Cantoni
(\cite{cantoni1975, cantoni1977, cantoni1977b, cantoni1980,
  cantoni1985}), by Cirelli and co-workers (\cite{cirelli1983,
  cirelli1984}), the more physically oriented
approach by Heslot \cite{Heslot1985a},  Bloch's paper
(\cite{Bloch1987}),
the work by Anandan\cite{springerlink:10.1007/BF00732829,Anandan1990},
and then Ashtekar and Schilling 
\cite{Ashtekar:1997p906}. There are several interesting works by
Brody and coworkers, \cite{Brody2001} being the 
one closer to the work presented here and also from Spera and coworkers
(\cite{Benvegnu2004,spera1993generalized}) .

\subsection{Representation of pure states}
The first step is to replace the Hilbert space ${\mathcal H}$ which models
the set of vector  states by a description in terms of real differential
manifolds.  Thus we replace  the Hilbert
space $\mathcal{H}$ with its realification $\mathcal{H}_\R:=M_Q$.
In this realification process the complex structure on
$\Hil$ will be represented by a tensor $J$ on $M_Q$ as we will see. We
assume that the dimension of the manifold $M_{Q}$ is equal to $2n$.

The natural identification is then  provided by choosing a basis
$\{|z_k\rangle\}$ in $\Hil$ and splitting the corresponding coordinates
into their real and imaginary parts:

$$
|\psi\rangle=\sum_k \psi_k |z_k\rangle \qquad \psi_k\to
\psi_k^R+i\psi_k^I 
$$

Then,
$$ \{ \psi_1, \cdots, \psi_n \} \in
\Hil \mapsto \{ \psi_1^R,\cdots, \psi_n^R,  \psi_1^I,\cdots \psi_n^I\}\equiv
(\Psi_R, \Psi_I)\in \Hil_\R. 
$$
Under this transformation, the Hermitian product becomes, for $\psi^1, \psi^2\in
\Hil$ 
$$
\langle(\Psi^1_R,\Psi^1_I), (\Psi^2_R,\Psi^2_I)\rangle =(\langle \Psi^1_R, \Psi^2_R\rangle + \langle \Psi^1_I,
\Psi^2_I\rangle )+i( \langle \Psi^1_R, \Psi^2_I\rangle -\langle \Psi^1_I,\Psi^2_R\rangle ).
$$

To consider $\Hil_\R$ just as a real differential manifold, the algebraic
structures available on $\Hil$ must be converted into tensor fields on
$\Hil_\R$. Consider first the tangent and cotangent bundles $T\Hil$
and $T^*\Hil$ and the following structures:
\begin{itemize}
\item The complex structure of $\mathcal{H}$ is translated into a
  tensor 
$$
J:M_Q\to M_Q,
$$
satisfying $J(\Psi_R, \Psi_I)=(-\Psi_I, \Psi_R)$ for any point
$(\Psi_R, \Psi_I)\in M_Q$. It is immediate to verify that
in this case
$$
J^2=-\mathbb{I}.
$$

\item  The linear structure available in
$M_Q$ is encoded in the vector field $\Delta$
$$
\Delta:M_Q\to TM_Q \quad \psi \mapsto (\psi, \psi).
$$
\item With every vector we can associate a vector field
$$
X_\psi:M_Q\to TM_Q \quad \phi \to (\phi, \psi)
$$
These vector fields are the infinitesimal generators of the vector
group $M_Q$ acting on itself. 
\item The Hermitian tensor $\langle\cdot, \cdot \rangle$ defined on
  the complex vector space $\mathcal{H}$, can be written in
  geometrical terms as
$$
\langle X_{\psi_1}, X_{\psi_2}\rangle(\phi) =\langle \psi_1, \psi_2\rangle.
$$
 On the ``real manifold'' the Hermitian scalar product may be  written as 
$$
\langle \psi_1, \psi_2\rangle=g(X_{\psi_1}, X_{\psi_2})+i\,\omega (X_{\psi_1}, X_{\psi_2}),
$$
where $g$ is now a symmetric tensor and $\omega $ a skew-symmetric
one. 
\end{itemize}

The properties of the Hermitian product ensure that:
\begin{itemize}
\item the symmetric tensor is positive definite and non-degenerate, and hence
  defines a Riemannian structure on the real vector manifold.
\item the skew-symmetric tensor is also non degenerate, and is closed with
  respect to the natural differential structure of the vector space. Hence, the
  tensor is a symplectic form (see also \cite{marmo1996})
\end{itemize}

As the inner product is sesquilinear, it satisfies
$$
\langle \psi_1, i\psi_2\rangle =i\langle \psi_1, \psi_2\rangle, \qquad
\langle i\psi_{1}, \psi_{2}\rangle =-i\langle
\psi_{1},\psi_{2}\rangle. 
$$
This  implies
$$
g(X_{\psi_1}, X_{\psi_2})=\omega (JX_{\psi_1}, X_{\psi_2}).
$$
We also have that $J^2=-\mathbb{I}$, and hence that the triple $(J, g, \omega )$
defines a K\"ahler structure (see \cite{cirelli1983}).
This implies, among other things, that the tensor 
$J$ generates both finite and infinitesimal transformations which are
orthogonal and symplectic.

The choice of the basis also allows us to introduce adapted coordinates for the
realified structure:
$$
\langle z_k, \psi\rangle =(q_k+ip_k)(\psi),
$$ 
and write the geometrical structures introduced above as:
$$
J=\partial_{p_k}\otimes dq_k-\partial_{q_k}\otimes dp_k 
\quad
g=dq_k\otimes dq_k+dp_k\otimes dp_k
\quad
\omega=dq_k\land dp_k
$$

\begin{note}
  If we represent the points of ${\mathcal H}$ by using complex
  coordinates we can write the Hermitian structure by means of
  $z_n=q^n+ip_n$:
$$
h=\sum_kd\bar z_k\otimes dz_k,
$$
where of course
$$
\langle X_{\psi_1}| X_{\psi_2}\rangle=h(X_{\psi_1}, X_{\psi_2}),
$$
the vector fields now being the corresponding ones on the complex
manifold.
\end{note}
In an analogous way we can consider a contravariant version of these
tensors.  The coordinate expressions with respect to the natural basis are:
\begin{itemize}
\item the Riemannian structure 
  \begin{equation}
G=\sum_{k=1}^n\left(\pdc{}{qk}\otimes \pdc{}{q^k}+\pdc{}{p_k}\otimes \pdc{}{p_k}\right),\label{eq:28}
\end{equation}
\item  the Poisson tensor
  \begin{equation}
\Omega=\sum_{k=1}^n\left(\pdc{}{q^k}\land \pdc{}{p_k}\right)\label{eq:29}
\end{equation}

\item while the complex structure has the form 
  \begin{equation}
J=\sum_{k=1}^n\left(\pdc{}{p_k}\otimes d{q^k}-\pdc{}{q^k}\otimes
  d{p_k}\right)\label{eq:30}
\end{equation}
\end{itemize}

\subsubsection{Example I: the Hilbert space of a two level quantum system}
\label{sec:example-i:-hilbert}
For a two levels system we will consider an orthonormal basis on
$\C^2$, say $\{|e_1\rangle,|e_2\rangle \} $. We introduce thus a set
of coordinates
$$
\langle e_j|\psi\rangle=z^j(\psi)=q^j(\psi)+ip_j(\psi) \qquad j=1,2.
$$

In the following we will use $z^j$ or $q^j$, $p_j$ omitting the
dependence in the state $\psi$ as it is usually done in differential
geometry.

The set of physical states is not equal to $\C^2$, since we have to
consider the equivalence relation given by the multiplication by a
complex number
i.e.
$$
\psi_1\sim \psi_2 \Leftrightarrow \psi_2=\lambda\psi_1 \qquad
\lambda\in \C_0=\C-\{ 0\}.
$$
And besides, the norm of the state must be equal to one. These two
properties can be encoded in the following diagram:
$$
\xymatrix{
\C^2-\{\vec 0\}:=\C^{2}_{0}\ar[rr]^\pi\ar[dr]& & S^2 \\
& S^3\ar[ur]_{\tau_H} &
}
$$
where $S^2$ and $S^3$ stand for the two and three dimensional spheres,
and the projection $\tau_H$ defines the Hopf fibration.  The projection
$\pi$ is associating each vector with the one-dimensional complex
vector space to which it belongs. Thus we see how this projection
factorizes through a projection onto $S^3$ and a further projection
given by the Hopf fibration, which is a $U(1)$--fibration.

The Hermitian inner product on $\C^2$ can be written in the
coordinates $z_1, z_2$ as
$$
\langle \psi|\psi \rangle=\bar z_jz^k \langle e_k|e_j\rangle=\bar z_jz^j.
$$
Equivalently we can write it in real coordinates $q,p$ and obtain:
$$
\langle\psi|\psi\rangle=p_1^2+p_2^2+(q^1)^2+(q^2)^2
$$

We can also obtain these tensors in contravariant form if we take as
starting point the Hilbert space ${\mathcal H}=\C^2$. If we repeat the steps
above, we obtain the two contravariant tensors:
$$
G=\pd{q^1}\otimes \pd{q^1}+\pd{p_1}\otimes \pd{p_1}+\pd{q^2}\otimes
\pd{q^2}+\pd{p_2}\otimes \pd{p_2};  \qquad
\Omega=\pd{q^1}\land \pd{p_{1}}+\pd{q^2}\land \pd{p_2}.
$$

Other   tensors encode the complex vector space structure of
$\Hil=\C^2$:
\begin{itemize}
\item the dilation vector field
  $\Delta=q^1\pd{q^1}+p_1\pd{p_1}+q^2\pd{q^2}+p_2\pd{p_2}$ ,
\item and the complex structure tensor $J=dp_1\otimes
  \pd{q^1}-dq^1\otimes \pd{p_1}+dp_2\otimes
  \pd{q^2}-dq^2\otimes \pd{p_2}$.
\end{itemize}

By combining both tensors, we can define the infinitesimal generator
of the multiplication by a phase:
$$
\Gamma=J(\Delta)=p_1\pd{q^1}-q^1\pd{p^1} + p_2\pd{q^2}-q^2\pd{p^2}.
$$
Thus we see how $\Delta$ is responsible for the quotienting from
$\C^2_0$ onto $S^3$, while $\Gamma$ is responsible for the Hopf
fibration $S^3\to S^2$.

\subsection{The complex projective space}

In the formulation as a real vector space, we can represent the
multiplication by a phase on the manifold $M_{Q}$  as  a transformation
whose infinitesimal generator is written as:  
\begin{equation}
  \label{eq:phase}
  \Gamma=\sum_{k}\left ( p_{k}\frac {\partial}{\partial q^{k}} -q^{k}
\frac {\partial}{\partial p_{k}}\right ).
\end{equation}

We can also consider another important vector field, which encodes the
linear space structure of the tangent bundle $TM_Q$. In order to avoid
singularities let us eliminate the zero section of the bundle $TM_Q$
and denote the resulting space by $T_0M_Q$. We remind the
reader that $M_Q$ is just the realification of a complex vector space
and, as such, we can encode its linear structure  in the dilation
vector field, which reads:
\begin{equation}
  \label{eq:15}
  \Delta:M_Q\to T_0M_Q; \qquad \psi\mapsto (\psi, \psi)
\end{equation}

In the coordinate system  $(q^k, p_j)$, it  takes the form
\begin{equation}
  \label{eq:14b}
  \Delta=q^k\frac{\partial}{\partial q^k}+p_k\frac{\partial}{\partial p_k}
\end{equation}

We are particularly interested in the relation of the vector fields
$\Delta$ and $\Gamma$.  In particular:
\begin{lemma}
  $\Delta$ and $\Gamma$ define a foliation on the manifold $M_Q$.
\end{lemma}

\begin{proof}
  It is simple to relate $\Delta$ with $\Gamma$ via the complex
  structure, in the form:
  \begin{equation}
    \label{eq:16}
    \Gamma=J (\Delta).
\end{equation}
Then it is straightforward to prove that both vector fields commute. 
  \end{proof}

We thus have an integrable distribution defined on the manifold
$M_Q$. We can thus define the corresponding quotient manifold
identifying the points which belong to the same orbit of the
generators $\Gamma$ and $\Delta$. Notice that, from the physical point
of view, this corresponds to the identification of points in the same
ray of the Hilbert space.

\begin{definition}
  The resulting quotient manifold, denoted as ${\mathcal P}$, defined as
  \begin{equation}
    \label{eq:17b}
    \pi:M_Q\to {\mathcal P}
  \end{equation}
 is  the
  \textbf{complex projective space} and its points represent the
  physical pure states of a quantum system. We will denote by
  $[\psi]$ the point in ${\mathcal P}$ which is the image by $\pi$ of a
  point $\psi\in M_Q$:
  \begin{equation}
    \label{eq:34}
   {\mathcal P}\ni  [\psi]:=\pi (\psi) \qquad \psi\in M_Q
  \end{equation}
\end{definition}

\section{The Ehrenfest picture for pure states}
Having introduced the necessary tools, let us proceed to describe the
Ehrenfest picture of Quantum Mechanics.  
As we saw in the introduction,  the key point consists in the
description of physical observables in terms of expectation value
functions. If we accept the point of view of formulating Quantum Mechanics in
terms of expectation value functions, we must also accept that they
are not defined on ${\mathcal H}$ but rather they are functions on the
arguments
$\frac{|\psi\rangle\langle\psi|}{\langle\psi|\psi\rangle}$. This means
that they are really functions defined 
on the Hilbert space ${\mathcal H}$ (or $M_{Q}$) which represent functions
on the projective space ${\mathcal P}$ corresponding to ${\mathcal H}$. This
change in the carrier space is not without consequences because now
the carrier space is not linear anymore: instead of a Hilbert space we
must consider a Hilbert manifold. This implies that we have to face
the problems of describing interference, superpositions of states  and
the composition of expectation value functions to replace the
multiplication rule of operators. 

Therefore, we must consider
functions on $M_{Q}$ 
which are constant along the fibers of the fibration
$\pi:M_Q\to {\mathcal P}$. Thus functions, meaningful from a physical
point of view, correspond to 
\begin{equation}
  \label{eq:18b}  
e_A(\psi)=\frac{\langle\psi|A\psi\rangle}{\langle\psi|\psi\rangle}
\end{equation}

These are functions on $M_Q$ which are in one-to-one
correspondence with the functions on the projective space ${\mathcal
  P}$ as pullback via the projection defined in
Eq. (\ref{eq:17b}). Obviously, they are no longer quadratic; but this
is a natural 
property taking into account that the projective space ${\mathcal P}$ has
lost the linear structure of $M_Q$ to become just a differential manifold.

\subsection{The spectral information}
One of the main aspects we must recover from the usual picture is the
spectrum of the operators. Indeed, in the usual descriptions of
Quantum Mechanics the spectrum of the observables encodes  most of the
information associated to the corresponding physical quantity. Thus, in the
Hilbert space description, given the observable $\hat A$, we associate
with it the basis of eigenvectors $\{ |v_{a}\rangle\}$ and the
corresponding eigenvalues:
$$
\hat A |v_{a}\rangle=a|v_{a}\rangle.
$$
Given a system in a state $|\psi(t)\rangle$, we also know that the
probability for a measurement of the observable $\hat  A$ at time $t$
to give the result $a$ is given by
$$
{\mathcal P}_{a}=|\langle v_{a}| \psi(t)\rangle|^{2}.
$$ 
How can we recover this information by using the expectation value
function $ e_{A}$ defined by Equation (\ref{eq:2})?  Notice that we are
considering it as a function defined the space of states $M_{Q}$
obtained by realification of the Hilbert space ${\mathcal H}$. In this
context., the information about the spectrum is recovered  easily from
the  set of critical points of the function. Indeed, it is immediate
to prove that the function $e_{A}$ has a critical point at each
eigenvector $|v_{a}\rangle$ while the value that it takes at those
points of $M_{Q}$ is precisely the eigenvalue of the operator $\hat
A$:
\begin{equation}
  \label{eq:3}
  \hat A |v_{a}\rangle=a|v_{a}\rangle \Leftrightarrow
  \begin{cases}
    de_{A}(|v_{a}\rangle)=0 \\
e_{A}(|v_{a}\rangle)=a
  \end{cases}
\end{equation}

%\textbf{Is this a note for us?}
%Perhaps, as a postulate, we should add that the state space of a
%composite system is the tensor product of the subsystem state spaces.

\subsection{The dynamics and the Poisson tensor}

We can also study the evolution of the system in the new picture. Let
us consider a pure state
\begin{equation}
  \label{eq:5}
  \rho_{\psi}(t)=\frac{|\psi(t)\rangle\langle\psi(t)|}{\langle\psi(t)|\psi(t)\rangle}
\end{equation}
and assume that the evolution is ruled by a one-parameter group of
transformations associated with a one-parameter group of unitary
transformations of the Hilbert space ${\mathcal H}$.  At the level of the
function $e_{A}$ we can write:
\begin{equation}
  \label{eq:6}
  \frac d{dt}\frac{\langle\psi|\hat
    A|\psi\rangle}{\langle\psi|\psi\rangle}=
\frac 1{\langle\psi(t)|\psi(t)} \left [ \left \langle \frac {d\psi(t)}{dt} |
\hat A \psi(t)\right \rangle +\left \langle  \psi(t)|
\hat A \frac {d\psi(t)}{dt} \right \rangle \right ];
\end{equation}
where we assume that the evolution preserves the norm of the state
$|\psi(t)\rangle$  and that $\hat A$ does not depend explicitely on
time. 

From Stone's theorem, we know that the unitary evolution on the
Hilbert space ${\mathcal H}$ is generated by a skew-hermitian generator in
the form:
\begin{equation}
  \label{eq:7}
  i\frac{d}{dt}|\psi(t)\rangle=\hat H |\psi(t)\rangle
\end{equation}
for some Hermitian operator $\hat H$.  If we introduce the commutator
of the operators $\hat H$ and $\hat A$ as
\begin{equation}
  \label{eq:8}
  [\hat H, \hat A]=i(\hat H\hat A-\hat A \hat H),
\end{equation}
Equation (\ref{eq:6}) can be written in terms of expectation value
functions as
\begin{equation}
  \label{eq:9}
 i \frac d{dt}e_{A}(\psi)=-\frac{\langle\psi|\hat H\hat
    A-\hat A \hat H|\psi\rangle}{\langle\psi|\psi\rangle}=ie_{[\hat H,
    \hat A]}(\psi).
\end{equation}
If we want to formulate completely the problem in terms of expectation
value functions, it makes sense to introduce an operation on this set
of functions, in particular a \textit{quantum Poisson bracket} defined
as
\begin{equation}
  \label{eq:10}
  \{ e_{H}, e_{A}\} :=e_{[\hat H, \hat A]}.
\end{equation}
We can extend this construction to the space of expectation value
functions by defining a bidifferential operator $\Omega_{{\mathcal P}}$ which
represents the Poisson tensor corresponding to the bracket above:
\begin{equation}
  \label{eq:14}
  \Omega_{{\mathcal P}} (de_{A}, de_{B}):=\{ e_{A}, e_{B}\} 
\end{equation}

Notice that these tensors, even if they are defined on the manifold
$M_{Q}$, are not the same as the tensor in  Eq 
(\ref{eq:29}). Indeed,  the functions are projectable under $\pi:M_Q\to {\mathcal P}$, but it
is simple to understand that the product under ($G$ and) $\Omega$ is not, since the tensors are
of degree -2, i.e., the Lie derivative of the tensors with 
respect to the dilation vector field $\Delta$ defined in
Eq. \eqref{eq:15} is 
$$
 {\mathcal L}_\Delta\Omega=-2\Omega.
$$
Thus, in order to make it projectable, we must rescale it by a factor
of degree two, for instance the square of norm of $|\psi\rangle$ which
is a central element

\begin{equation}
\label{eq:21b}
\{ e_A, e_B\}_{\mathcal P}:=\Omega_{\mathcal P}(de_A,
de_B)=\langle\psi|\psi\rangle \{e_A, e_B\}
\end{equation}

\subsection{Indetermination relations and the symmetric structure}
Another aspect that we have to take into account in our description is
the formulation in terms of expectation value funtions of another important aspect of Quantum Mechanics
as it is the indetermination relations.  This introduces the necessity
of bringing the tensor $G_{{\mathcal P}}$, defined in a similar way as we
did for $\Omega_{{\mathcal P}}$,  into play.

It is simple to verify that, given an operator $\hat A$, its squared
uncertainty (its variance) can be obtained from the expectation value
function $e_{A}$ as:
\begin{equation}
  \label{eq:11}
  (\Delta A)^{2}=\langle de_{A}|de_{A}\rangle=e_{A^{2}}-(e_{A})^{2}
\end{equation}
where we represent by $d$ the exterior differential in ${\mathcal H}$ and
we use the extension of the Hermitian structure to the differential
one-forms. 

But if we want to implement the variance directly at the level of the
projective space ${\mathcal P}$, we need to introduce a differential
operator which encodes the symmetric product of operators at the level
of the expectation value functions. We consider then the corresponding
tensor which precisely coincides with the tensor $G_{{\mathcal P}}$:
\begin{equation}
  \label{eq:13}
  G_{{\mathcal P}}(de_{A}, de_{B}):=e_{AB+BA}-e_{A}e_{B}
\end{equation}

This symmetric tensor $G_{{\mathcal P}}$ together with the tensor $\Omega_{{\mathcal P}}$ defined
in Eq. (\ref{eq:14}) endow the projective space with a Hermitian
bidifferential operator.  We shall work explitely the various aspects
of the construction by means of a very simple example: a two level
system defined on a Hilbert space ${\mathcal H}=\mathbb{C}^{2}$.

\subsubsection{Example II: the projective space for a two level quantum system}

Extending the example presented in Section \ref{sec:example-i:-hilbert}, we
can consider now the corresponding projective space and the
corresponding tensors.
 It is important to remark that while forms can not be
projected, contravariant tensor fields can. This is the reason why we
introduced the contravariant tensors $\Lambda$ and $G$. Thus by considering 
$$
G=\pd{q^1}\otimes \pd{q^1}+\pd{p_1}\otimes \pd{p_1}+\pd{q^2}\otimes
\pd{q^2}+\pd{p_2}\otimes \pd{p_2} ,
$$
we can consider the projection of the tensor. As it happens with the
Poisson tensor, it is immediate to
understand that such a tensor can not be projected directly, since it
is of degree two with respect to the dilations, i.e.
$$
{\mathcal L}_{\Delta}G=-2G
$$
 Therefore, we have to
consider a conformal factor and define (see \cite{Ercolessi2010}):
\begin{multline}
  \label{eq:project_g}
 G_{\mathcal P}=\langle \psi|\psi\rangle G-\Gamma\otimes \Gamma
 -\Delta\otimes \Delta=\\
=((q^1)^2+(q^2)^2+p_1^2+p_2^2)\left (
\pd{q^1}\otimes \pd{q^1}+\pd{p_1}\otimes \pd{p_1}+\pd{q^2}\otimes
\pd{q^2}+\pd{p_2}\otimes \pd{p_2}
  \right) -\\
\sum_{lm}\left ( p_{l}\frac{\partial}{\partial q^{l}}-
  q^{l}\frac{\partial}{\partial p_{l}}\right )
\otimes \left ( p_{m}\frac{\partial}{\partial q^{m}}-
  q^{m}\frac{\partial}{\partial p_{m}}\right ) - \\
\sum_{lm}\left ( q^{l}q^{m} \frac{\partial}{\partial q^{l}}\otimes
\frac{\partial}{\partial q^{m}} +   p_{l}p_{m}\frac{\partial}{\partial p_{l}}\otimes
\frac{\partial}{\partial p_{m}} \right )
\end{multline}

Analogously we can write the coordinate expression of the tensor $\Omega_{\mathcal P}$:
\begin{multline}
  \label{eq:lambdaH}
  \Omega_{\mathcal P}=\langle \psi |\psi\rangle \Omega-\Gamma\otimes
  \Delta  -\Delta\otimes \Gamma =
((q^1)^2+(q^2)^2+p_1^2+p_2^2)\left (\pd{q^1}\land
    \pd{p^1}+ \pd{q^2}\land \pd{p^2}\right ) -\\
-
\sum_{lm}\left ( p_{l}\frac{\partial}{\partial p_{l}}+
  q^{l}\frac{\partial}{\partial q^{l}}\right )
\otimes \left ( p_{m}\frac{\partial}{\partial q^{m}}-
  q^{m}\frac{\partial}{\partial p_{m}}\right )
-\\
\sum_{lm}\left ( p_{l}\frac{\partial}{\partial q^{l}}-
  q^{l}\frac{\partial}{\partial p_{l}}\right )
\otimes\left ( p_{l}\frac{\partial}{\partial p_{l}}+
  q^{l}\frac{\partial}{\partial q^{l}}\right )
\end{multline}

\subsection{Uncertainty  relations}

From the definition of the variance,  we can write the Robertson
version of the  uncetainty relations (see \cite{Robertson1929}) in a simple
form:
\begin{equation}
  \label{eq:12}
\Delta_{\psi} A \Delta_{\psi} B \geq \frac
  14 \langle \psi |[A, B]\psi\rangle^{2}
\end{equation}

We are going to obtain this well known expression within our Ehrenfest
picture. Let us consider an
arbitrary operator $F$ on ${\mathcal H}$. It is immediate that
$$
\langle \psi |F^{\dagger }F \psi\rangle \geq 0 ;\qquad \forall |\psi
\rangle\in {\mathcal H} .
$$
For simplicity, we will restrict the set of states to the sphere of
normalized states ${\mathcal S}$, i.e. we shall consider the inequality:

\begin{equation}
\langle \psi |F^{\dagger }F \psi\rangle \geq 0 ;\qquad \forall |\psi
\rangle\in {\mathcal S} .\label{eq:48}
\end{equation}

Let us  choose $F$ to be  the operator which is a complex linear
combination of two hermitian observables: 
\begin{equation}
  \label{eq:47}
  F=(A - \langle A\rangle_{\psi}\mathbb{I})+i\alpha (B-\langle
  B\rangle_{\psi}\mathbb{I}), \qquad \alpha \in \mathbb{R}
\end{equation}
where $\langle A\rangle_{\psi}$ and $\langle B\rangle_{\psi}$ represent the
expectation value of each observable  in a given state.

In this situation, the inequality (\ref{eq:48}), as a polynomial in
$\alpha$, corresponds to: 
\begin{equation}
  \label{eq:46}
  \alpha^{2}\left ( e_{B^{2}}(\psi) -e_{B}(\psi)^{2} \right )+ \alpha e_{[A, B]}(\psi)
  +  \left ( e_{A^{2}}(\psi) -e_{A}(\psi)^{2} \right ) \geq 0,
\end{equation}
where $[A, B]=i (AB-BA)$ to make it an inner operation in the set of
Hermitian operators.

The condition must hold for any value of $\alpha$ and hence we obtain
a condition on the roots, that can not be real. Then, we obtain:
\begin{equation}
  \label{eq:49}
  4 \left ( e_{B^{2}}(\psi) -e_{B}(\psi)^{2} \right )\left (
    e_{A^{2}}(\psi) -e_{A}(\psi)^{2} \right )  -e_{[A,
    B]}(\psi)^{2}\geq 0,
\end{equation}
or equivalently:
\begin{equation}
  \label{eq:50}
  \left ( e_{B^{2}}(\psi) -e_{B}(\psi)^{2} \right )\left (
    e_{A^{2}}(\psi) -e_{A}(\psi)^{2} \right )\geq \frac 14 e_{[A,B]}(\psi)^{2}.
\end{equation}

In this expression we recognize the usual formulation of the
uncertainty relation for two arbitrary operators if we write:
\begin{equation}
  \label{eq:51}
  (\Delta A)_{\psi}=\left (
    e_{A^{2}}(\psi) -e_{A}(\psi)^{2} \right )
\end{equation}
in this language, the relation becomes:
\begin{equation}
\Delta_{\psi}A \Delta_{\psi}B\geq \frac 14e_{[A,B]}(\psi)^{2}. 
 \label{eq:36}
  \end{equation}

This new expression allows us to write uncertainty relations by using
only tensors $G_{\mathcal P}$ and $\Omega_{{\mathcal P}}$:
\begin{equation}
  \label{eq:37}
  G_{{\mathcal P}}(de_{A}, de_{A})G_{{\mathcal P}}(de_{B}, de_{B}) \geq
  \frac 14\left (\Omega_{{\mathcal P}}(de_{A}, de_{B})\right )^{2}.
\end{equation}
%Thus we verify that the uncertainty relations depend only on the
%symmetric tensor $G_{{\mathcal P}}$.

%\textbf{The other two operations: recover the associative product and
  %the symmetric one}

It is also possible to provide an analogous formulation for
Schrödinger uncertainty relations (see \cite{schrodinger1930heisenbergschen}).
Consider the same Hermitian operators
$A$ and $B$ as above, and consider the expectation value of the
product
$$
K=K_{A}K_{B}= (A - \langle A\rangle_{\psi}\mathbb{I})(B - \langle B\rangle_{\psi}\mathbb{I})
$$
where $K_{A}=(A - \langle A\rangle_{\psi}\mathbb{I}$ and $K_{B}=(B -
\langle B\rangle_{\psi}\mathbb{I})$. 
 From Schwartz inequality, we can write:

 \begin{equation}
|\langle \psi| K\psi\rangle|^{2}=|\langle \psi
|K_{A}K_{B}|\psi\rangle|^{2}\leq \langle \psi|
K_{A}^{2}|\psi\rangle\langle \psi|K_{B}^{2}|\psi\rangle\label{eq:52}
\end{equation}

Now, we can replace the product $K_{A}K_{B}$ by:
\begin{multline*}
K_{A}K_{B}=\frac 12 (K_{A}K_{B}+K_{B}K_{A})+\frac 12
(K_{A}K_{B}-K_{B}K_{A})= \\ \frac 12 (K_{A}K_{B}+K_{B}K_{A})+\frac i2
[K_{A}, K_{B}].
\end{multline*}
We can write then:
$$
|\langle \psi |K_A K_{B}\psi\rangle|^{2}=\frac 14 \left ( \langle |
  \psi(K_{A}K_{B}+K_{B}K_{A})|\psi\rangle \right )^{2}+\frac 14
(\langle \psi |[K_{A}, K_{B}]\psi\rangle)^{2}.
$$
It is straightforward to verify that
$$
\langle \psi |[K_{A}, K_{B}]\psi\rangle=\langle \psi|[A,B]|\psi\rangle,
$$
because the identity operators trivially commute; and
$$
 \langle 
  \psi |(K_{A}K_{B}+K_{B}K_{A})|\psi\rangle =
\langle \psi | (AB+BA)|\psi\rangle-\langle \psi|A\psi\rangle\langle\psi|B|\psi\rangle.
$$

Analogously
$$
\langle \psi|
K_{A}^{2}|\psi\rangle=\langle \psi|A^{2}|\psi\rangle-\langle
\psi|A|\psi\rangle^{2}; 
\qquad
K_{B}^{2}|\psi\rangle=\langle \psi|B^{2}|\psi\rangle-\langle
\psi|B|\psi\rangle^{2}.
$$

We can then write Equation (\ref{eq:52}) as:
\begin{multline}
  \label{eq:53}
  (\langle \psi|A^{2}|\psi\rangle-\langle
\psi|A|\psi\rangle^{2})(\langle \psi|B^{2}|\psi\rangle-\langle
\psi|B|\psi\rangle^{2})\geq  \\ \frac 14 \left (\langle \psi |
  (AB+BA)|\psi\rangle-\langle
  \psi|A\psi\rangle\langle\psi|B|\psi\rangle \right )^{2} + \frac 14 \langle \psi|[A,B]|\psi\rangle^{2}
\end{multline}
or, analogously,
\begin{equation}
  \label{eq:54}
  (e_{A^{2}}(\psi)-e_{A}(\psi)^{2}) (e_{B^{2}}(\psi)-e_{B}(\psi)^{2})
  - \frac 14(e_{(A\circ B)}(\psi)-e_{A}(\psi)e_{B}(\psi))^{2} \geq \frac 14
  \langle \psi|[A,B]|\psi\rangle^{2} ,
\end{equation}
where $A\circ B=(AB+BA)$.

This expression, which is known as Schrödinger uncertainty relation,
can also be written in terms of the tensors $G_{{\mathcal P}}$ and
$\Omega_{{\mathcal P}}$ in the form:
\begin{equation}
  \label{eq:55}
  G_{{\mathcal P}}(de_{A}, de_{A}) G_{{\mathcal P}}(de_{B}, de_{B})-\frac 14
  G_{{\mathcal P}}(de_{A}, de_{B})^{2}\geq \frac 14 \Omega_{{\mathcal P}}(de_{A}, de_{B}).
\end{equation}
\section{Ehrenfest picture for mixed states: the two levels system}

In the previous section we have been able to construct, by using the
tensors which encode the Hermitian structure of the Hilbert space of
states, a formulation of Quantum Mechanics where the physical
observables are represented by the expectation value functions
associated to pure states of the physical system.
The next step is to consider the generalization to the case of mixed
states. For the sake of simplicity, we shall consider only the case of
the two level system that we have analyzed so far. In order to do
that, we shall begin by reformulating the construction above in terms
of rank-one projectors, and later we will be able to extend this new
framework to include arbitrary mixed states.

\subsection{Reformulation of pure states}
We know that the rank-one projectors defined on ${\mathcal H}$ are in
one-to-one correspondence with the points of the projective space
${\mathcal P}$. Indeed, we can write:
\begin{equation}
  \label{eq:18}
  \begin{pmatrix}
    \bar z_1 z_1 & \bar z_1 z_2 \\
    \bar z_2z_1 & \bar z_2z_2
  \end{pmatrix}=\rho_\psi=
y_0\sigma_0+y_1\sigma_1+y_2\sigma_2+y_3\sigma_3,
\end{equation}
where we have to impose that
$$
\mathrm{Tr} \rho_\psi=1 \Rightarrow y_0=\frac 12
$$
and
$$
\mathrm{Tr}\rho_\psi^2=\mathrm{Tr}\rho_\psi \Rightarrow
y_0^2+y_1^2+y_2^2+y_3^2=\frac 12 \Rightarrow 
y_1^2+y_2^2+y_3^2=\frac 1 4
$$
We conclude thus that the set of rank-one projectors on vectors of the
Hilbert space ${\mathcal H}=\mathbb{C}^{2}$, which is in one-to-one correspondence with the points
of the projective space ${\mathcal P}=\mathbb{CP}^{1}$ is diffeomorphic
to the two dimensional sphere $S^{2}$.

The coordinates $\{y_{0}, y_{1}, y_{2}, y_{3}\}$ can  be obtained 
from the properties of the 
Pauli matrices:

$$
y_{0}=\frac 12 \mathrm{Tr}(\sigma_{0}\rho_{\psi})=\frac 12
$$
 and
$$
y_{j}=\frac 12 \mathrm{Tr}(\sigma_{j}\rho_{\psi}).
$$

Expectation value functions can be defined for any Hermitian operator
in the form 
$$
A=a_{0}\sigma_{0}+\alpha_{1}\sigma_{1}+a_{2}\sigma_{2}+a_{3}\sigma_{3},
$$
as the evaluation of the operator on the state $\rho$:
\begin{equation}
  \label{eq:19}
  e_{A}(\rho)=\mathrm{Tr}(A\rho)=\mathrm{Tr} [A(\frac 12
  \sigma_{0}+y_{1}\sigma_{1}+y_{2}\sigma_{2}+ y_{3}\sigma_{3})]=a_{0}+2(a_{1}y_{1}+a_{2}y_{2}+a_{3}y_{3})
\end{equation}

The embedding of the set of pure states in this form, forces us to use
Lagrange multipliers in order to, for instance, determine the set of
critical points of the function $e_{A}(\rho)$ on the set of pure states. Thus, we compute the
extremal points of the function 
$$
f_{A}(\lambda,
\rho)=e_{A}(\rho)-\lambda(y_{1}^{2}+y_{2}^{2}+y_{3}^{2}-\frac 14). 
$$
We obtain thus:
$$
df_{A}(y^{\star})=0 \rightarrow
\begin{cases}
  2[(a_{1}-\lambda y_{1}^{\star} )dy_{1}+(a_{2}-\lambda
  y_{2}^{\star})dy_{2}+(a_{3}-\lambda y_{3}^{\star})dy_{3}]=0
  \\
(y_{1}^{\star})^{2}+(y_{2}^{\star})^{2}+(y_{3}^{\star})^{2}-\frac 14=0
\end{cases}
$$
This implies that:
$$
y_{1}^{\star}=\frac {a_{1}}{\lambda}; \quad y_{2}^{\star}=\frac {a_{2}}{\lambda};
\quad y_{3}^{\star}=\frac {a_{3}}{\lambda}
$$
and thus:
$$
a_{1}^{2}+a_{2}^{2}+a_{3}^{2}=\frac 14 \lambda^{2}\Rightarrow
\lambda=\pm \sqrt{4(a_{1}^{2}+a_{2}^{2}+a_{3}^{2}) }.
$$
The corresponding eigenvalue is obtained as:
\begin{equation}
  \label{eq:20}
  e_{A}(\rho^{\star})=a_{0}\pm\frac{2}{\sqrt{
      4(a_{1}^{2}+a_{2}^{2}+a_{3}^{2}) }}\left (
    a_{1}^{2}+a_{2}^{2}+a_{3}^{2}\right ) = a_{0}\pm \sqrt{(a_{1}^{2}+a_{2}^{2}+a_{3}^{2})}
  \end{equation}

If we consider a second operator
\begin{equation}
  \label{eq:21}
  B=b_{0}\sigma_{0}+b_{1}\sigma_{1}+b_{2}\sigma_{2}+b_{3}\sigma_{3}
\end{equation}
and the corresponding function $e_{B}(\rho)$, we can evaluate the
commutator and the skew-commutator:
\begin{equation}
  \label{eq:22}
  [A,B]:=i(AB-BA)
\end{equation}
\begin{equation}
  \label{eq:23}
  A\circ B=(AB+BA)
\end{equation}
and the corresponding functions:
\begin{equation}
  \label{eq:24}
  e_{[A,B]}(\rho)=4(a_{3}b_{2}-a_{2}b_{3})y_{1}+4(a_{1}b_{3}-a_{3}b_{1})y_{2}+4(a_{2}b_{1}-a_{1}b_{2})y_{3}; 
\end{equation}
and
\begin{equation}
  \label{eq:25}
  e_{A\circ
    B}(\rho)=4(a_{0}b_{0}+a_{1}b_{1}+a_{2}b_{2}+a_{3}b_{3})y_0+
4(a_1b_0+a_0b_1)y_1+4(a_2b_0+a_0b_2)y_2+4(a_3b_0+a_0b_3)y_3.
\end{equation}
From both epxressions we read therefore the coordinate expression of
the tensors representing the operations at the level of the
expectation value functions:
\begin{equation}
  \label{eq:26}
  G(\rho)=4\left ( y_0 \sum_{j=0}^4\frac{\partial}{\partial
      y_j}\otimes \frac{\partial}{\partial y_j}
    +\sum_{j=1}^3\left (y_j\frac{\partial}{\partial y_j} \otimes
    \frac{\partial}{\partial y_0}  +y_j\frac{\partial}{\partial y_0} \otimes
    \frac{\partial}{\partial y_j} \right )\right ) 
\end{equation}
\begin{equation}
  \label{eq:27}
  \Lambda(\rho)=\sum_{jkl=1}^3\epsilon^{jkl}y_j
  \frac{\partial}{\partial y_k}\wedge \frac{\partial}{\partial y_l} 
\end{equation}

\subsection{Ehrenfest picture of mixed states}

The set of mixed states ${\mathcal D}$ can be defined by relaxing the
condition defining the projector property of $\rho$, i.e.,

\begin{equation}
  \label{eq:17}
  {\mathcal D}=\{ \rho \in \mathfrak{u}^{*}({\mathcal H}) | \mathrm{Tr} \rho=1\}
\end{equation}

By using the same coordinates, we end up with the following coordinate
description:
\begin{equation}
  \label{eq:31}
  {\mathcal D}=\{ (y_{1}, y_{2}, y_{3})\in \mathbb{R}^{3} |
  y_{1}^{2}+y_{2}^{2}+y_{3}^{2}\leq \frac 34\} 
\end{equation}

Once this is specified, the construction is completely analogous to
the previous case. Expectation value functions are again defined as
\begin{equation}
  \label{eq:32}
  e_{A}(\rho)=\mathrm{Tr}(\rho A); \qquad \rho\in {\mathcal D}, 
\end{equation}
while dynamics is defined through the Poisson tensor $\Omega$ and
other physical properties as the uncertainty relations are encoded
again in the tensor $G$.

\section{Interference in the Ehrenfest picture}
\subsection{Describing interference on ${\mathcal D}$}
In the sections above we have been able to address some of the
problems which arise when we formulate Quantum Mechanics in the
Ehrenfest picture. In particular we have been able to define the
dynamics and the uncertainty relations in terms of geometric objects
which are defined either on the projective space ${\mathcal P}$ associated
to the Hilbert space ${\mathcal H}$ or on the space of density operators
${\mathcal D}$. 

Our final exercise addresses the problem of formulating interference
phenomena within the framework. Consider then that we are given a pair
of pure states $\rho_{1}$ and $\rho_{2}$ and that we want to find a
method to combine them into a new pure state. Thus we must be able to define a
procedure to combine two  rank-one projectors into a new one. In order
to do that, we need to fix a fiducial projector $P_{0}$ which will
allow us to consider the relative phases which are essential to
describe the interference process, diffraction or also the composition
of light polarization. 

Consider thus, given $\rho_{1}, \rho_{2}$ and $P_{0}$ as above, and
$p_{1}, p_{2}\in [0,1]$ with $p_{1}+p_{2}=1$,  the
operator
\begin{equation}
  \label{eq:38}
  \rho=p_{1}\rho_{1}+p_{2}\rho_{2}+\frac{\sqrt{p_{1}p_{2}}}{\mathrm{Tr}(\rho_{1}P_{0}\rho_{2}P_{0})} 
(\rho_{1}P_{0}\rho_{2}+h.c.).
\end{equation}
It is straightforward to prove that, with the conditions above, $\rho$
is also a pure state, i.e.,
$$
\rho^{2}= \rho; \qquad \mathrm{Tr} \rho=1.
$$
Besides, it is simple to check that
$$
\rho_{1}\rho\rho_{1}=p_{1}\rho_{1} ; \qquad \rho_{2}\rho\rho_{2}=p_{2}\rho_{2}.
$$

Notice that this composition law may be considered to provide us with
a purification of the state $\rho=p_{1}\rho_{1}+p_{2}\rho_{2}$ when
the two pure states are orthogonal,  and the projection on $P_{0 }$ of
both is different from zero, i.e.
$$
\rho_{1}\rho_{2}=0; \qquad P_{0}\rho_{1}\neq 0; \qquad
P_{0}\rho_{2}\neq 0.
$$

When the two pure states $\rho_{1}$ and $\rho_{2}$ are not orthogonal 

\begin{equation}
  \label{eq:39}
   \rho=p_{1}\rho_{1}+p_{2}\rho_{2}+\frac{\sqrt{p_{1}p_{2}}}{\mathrm{Tr}(\rho_{1}P_{0}\rho_{2}P_{0})} 
(\rho_{1}P_{0}\rho_{2}+h.c.)W^{{-1}},
\end{equation}
where 
\begin{equation}
  \label{eq:40}
  W=1+\frac{\sqrt{p_{1}p_{2}}}{\mathrm{Tr}(\rho_{1}P_{0}\rho_{2}P_{0})} 
\mathrm{Re}(\rho_{1}P_{0}\rho_{2}+h.c.)
\end{equation}

To summarize the construction, we can say that having chosen a
fiducial projector $P_{0}$, we are able to define a composition
procedure of pure states which is an inner operation.
For further details on this issue, we refer the interested reader to
\cite{Manko2002} and \cite{Ercolessi2010}. In the following we will
study the geometrical meaning of the superposition procedure.

The projector $P_{0}$ we introduced plays the rôle of the Pancharatnam
connection (see \cite{Pancha1956, Pancharatnam1956} )
 which encodes the geometrical description of Berry
phase (see \cite{Berry1984a,Simon1983, Anandan1992,Aharonov1987}. Indeed, it
provides a way to lift the physical states from the 
complex projective space into the Hilbert space.  Once we have vectors
in the Hilbert space, it is straightforward to evaluate the transition
probability from one vector to another. In brief, we choose the
fiducial projector as the rank one projector on a vector
$|\psi_{0}\rangle\in {\mathcal H}$:
\begin{equation}
  \label{eq:41}
  P_{0}=\frac{|\psi_{0}\rangle\langle\psi_{0}|}{\langle \psi_{0}|\psi_{0}\rangle},
\end{equation}
and then we are able to associate to the rank-one projector $\rho_{1}$
(respectively $\rho_{2}$)
the vector
\begin{equation}
  \label{eq:42}
  |\psi_{1}\rangle
  =\frac{1}{\sqrt{\langle\psi_{0}|\psi_{0}\rangle}}\rho_{1}|\psi_{0}\rangle;
  \qquad 
|\psi_{2}\rangle=\frac{1}{\sqrt{\langle\psi_{0}|\psi_{0}\rangle}}\rho_{2}|\psi_{0}\rangle.
\end{equation}
The transition probability between states $\rho_{1}$ and $\rho_{2}$
can be defined as
\begin{equation}
  \label{eq:43}
  {\mathcal P}_{{1-2}}=|\langle \psi_{1}|\psi_{2}\rangle|^{2}=\frac{\langle
  \psi_{0}|\rho_{1}\rho_{2}|\psi_{0}\rangle}{\langle \psi_{0}|\psi_{0}\rangle} =\mathrm{Tr}(\rho_{1}P_{0}\rho_{2}).
\end{equation}

Then, with that result, the composition of the two states $|\psi_{1}\rangle$ and
$|\psi_{2}\rangle$ allows us to describe  properly the interference phenomena.

\subsection{Example: the case of a two-level system}
If we consider now the example of the two level system, we know that
the projective space is diffeomorphic to the two dimensional sphere
$S^{2}$.  If we want to define a linear structure on the sphere with
the help of the projector $P_{0}$, we must exclude the points which
are orthogonal to $|\psi_{0}\rangle$.  We consider thus as set of
representable states 
$$
V=S^{2}-\{ P_{0}^{\perp}\}.
$$
The linear structure that we are defining on $V$ is analogous to the
one we obtain if we consider a point $s_{0}\in S^{2}$, the
corresponding tangent space $T_{s_{0}}S^{2}$; and with the help of any
second order differential equation on $S^{2}$  we define the time-one
map to the sphere by means of the flow $\Phi:\mathbb{R}\times
TS^{2}\to TS^{2}$, i.e.:
$$
\Phi(t, s_{0}, v)|_{t=1}; \qquad v\in T_{s_{0}}S^{2}.
$$
Then, to any point in $s\in S^{2}$ (except the focal point of $s_{0}$, this
is the reason to exclude $P_{0}^{\perp}$), we can associate a vector
$v$ in the tangent space $T_{s_{0}}S^{2}$:
\begin{equation}
  \label{eq:44}
  s=\Phi(t=1, s_{0}, v); \qquad S^{2}\ni s\leftrightarrow  v\in T_{s_{0}}S^{2}.
\end{equation}
It is natural to consider a particular case for this second order
differential equation, as is the geodetic motion defined on the
sphere. Thus the mapping $\Phi$ corresponds to the exponential
mapping associated with the Riemannian structure of the projective space.

Then, given two arbitrary points $s_{1}, s_{2}\in V$, we can associate
a third element:
\begin{equation}
  \label{eq:45}
  \begin{cases}
    s_{1}=\Phi(t=1, s_{0}, v_{1}) \\
s_{2}=\Phi(t=1, s_{0}, v_{2})
  \end{cases}
 \Rightarrow  s_{1}\star s_{2}=\Phi(t=1, s_{0}, v_{1}+v_{2})
\end{equation}

Once the linear structure has been implemented, the description of
interference phenomena reduces to incoporate the scalar product of the
physical states. Within our new framework, this can be 
accomplished by using the tensors $G_{{\mathcal P}}$ and $\Omega_{{\mathcal
    P}}$ defined in Equations (\ref{eq:project_g}) and
(\ref{eq:lambdaH}). We can define thus:
\begin{equation}
  \label{eq:33}
  (s_{1}, s_{2}):=G_{{\mathcal P}}(v_{1}, v_{2})+i \Omega_{{\mathcal
      P}}(v_{1}, v_{2}).
\end{equation}

Summarizing:  on the projective space ${\mathcal P}\sim S^{2}$ minus the
focal point to a fiducial point $s_{0}$, we can induce a vector space
structure by using the linear structure of the vector space
$T_{s_{0}}S^{2}$.  Clearly, if two different fiducial points $s_{0},
s_{0}'$ are used, we induce two alternative linear structures on the
space $S^{2}$ excluding the two focal points.  The transition
function from one linear structure to the other defines a nonlinear
map. This construction shows the importance of alternative linear
structures in Quantum Mechanics and how the usual linear structure of
the Hilbert space formalism is an ingredient chosen by the observer,
via the fiducial point used as reference for the interference
phenomena.

%CHECK AND FINISH

\section*{Acknowledgments}
G.M. would like to acknowledge the support provided by the Santander/UCIIIM
Chair of Excellence programme 2011-2012. The work of J. C-G has been
partially supported by Grants DGA 24/1, MICINN MTM2012-33575  and
UZ2012CIE-06

%\bibliographystyle{plain}
%\bibliography{/Users/jesus/Dropbox/TeX/bibliography/library}

\begin{thebibliography}{10}

\bibitem{Aharonov1987}
Y~Aharonov and J~Anandan.
\newblock {Phase Change during a Cyclic Quantum Evolution}.
\newblock {\em Physical Review Letters}, 58(16):1593--1596, 1987.

\bibitem{springerlink:10.1007/BF00732829}
J~Anandan.
\newblock {A geometric approach to quantum mechanics}.
\newblock {\em Foundations of Physics}, 21(11):1265--1284, 1991.

\bibitem{Anandan1992}
J.~Anandan.
\newblock {The geometric phase}.
\newblock {\em Nature}, 360:308--313, 1992.

\bibitem{Anandan1990}
J.~Anandan and Y.~Aharonov.
\newblock {Geometry of quantum evolution}.
\newblock {\em Physical Review Letters}, 65(14):1697--1700, 1990.

\bibitem{Aniello}
P.~Aniello, J.~Clemente-Gallardo, G.~Marmo, and G.~F. Volkert.
\newblock {Classical tensors and quantum entanglement I: pure states}.
\newblock {\em International Journal of Geometric Methods in Modern Physics},
  07(03):485, November 2010.

\bibitem{Aniello2011}
P.~Aniello, J.~Clemente-Gallardo, G.~Marmo, and G.~F. Volkert.
\newblock {Classical Tensors and Quantum Entanglement II: Mixed States}.
\newblock {\em International Journal of Geometric Methods in Modern Physics},
  8(4):853--883, 2011.

\bibitem{Ashtekar:1997p906}
A.~Ashtekar and T.~A. Schilling.
\newblock {Geometrical Formulation of Quantum Mechanics}.
\newblock {\em On Einstein’s Path: Essays in Honor of Engelbert Schucking},
  pages 23--65, 1997.

\bibitem{Benvegnu2004}
A.~Benvegn\`{u}, N.~Sansonetto, and M.~Spera.
\newblock {Remarks on geometric quantum mechanics}.
\newblock {\em Journal of Geometry and Physics}, 51(2):229--243, June 2004.

\bibitem{Berry1984a}
M.~V. Berry.
\newblock {Quantal Phase Factors Accompanying Adiabatic Changes}.
\newblock {\em Proceedings of the Royal Society A: Mathematical, Physical and
  Engineering Sciences}, 392(1802):45--57, March 1984.

\bibitem{Bloch1987}
A.~: Bloch.
\newblock {An infinite-dimensional hamiltonian system on projective hilbert
  space}.
\newblock {\em Transactions AMS}, 302(2):787--796, 1987.

\bibitem{Brody2001}
D~Brody and L.~P. Hughston.
\newblock {Geometric quantum mechanics}.
\newblock {\em Journal of Geometry and Physics}, 38(1):19--53, April 2001.

\bibitem{cantoni1975}
V~Cantoni.
\newblock {Generalized “transition probability”}.
\newblock {\em Comm. Math. Phys.}, 44:125--128, 1975.

\bibitem{cantoni1977}
V~Cantoni.
\newblock {Intrinsic geometry of the quantum-mechanical phase space,Hamiltonian
  systems and Correspondence Principle}.
\newblock {\em Rend. Accd. Naz. Lincei}, 62:628--636, 1977.

\bibitem{cantoni1977b}
V~Cantoni.
\newblock {The Riemannian structure on the space of quantum-like systems}.
\newblock {\em Comm. Math. Phys.}, 56:189--193, 1977.

\bibitem{cantoni1980}
V~Cantoni.
\newblock {Geometric aspects of Quantum Systems}.
\newblock {\em Rend. Sem. Mat. Fis. Milano}, 48:35--42, 1980.

\bibitem{cantoni1985}
V~Cantoni.
\newblock {Superposition of physical states:a metric viewpoint}.
\newblock {\em Helv. Phys. Acta}, 58:956--968, 1985.

\bibitem{Carinena2007b}
J.~F. Cari\~{n}ena, J.~Clemente-Gallardo, and G.  Marmo.
\newblock {Geometrization of quantum mechanics}.
\newblock {\em Theoretical and Mathematical Physics}, 152(1):894--903, July
  2007.

\bibitem{Carinena2007a}
J.F. Carinena, J.~Clemente-Gallardo, and G. Marmo.
\newblock {Introduction to Quantum Mechanics and the Quantum-Classical
  transition}.
\newblock In D~Iglesias, J.~C. Marrero, E.~Padr\'{o}n, and D.~Sosa, editors,
  {\em Proc. XV International Workshop on Geometry and Physics}, number~1,
  pages 3--42. Publi. RSME, July 2007.

\bibitem{cirelli1983}
R~Cirelli, P~Lanzavecchia, and A~Mani\`{a}.
\newblock {Normal pure states of the von Neumann algebra of bounded operator as
  K\"{a}hler manifold}.
\newblock {\em J. Phys. A: Math Gen}, 15:3829--3835, 1983.

\bibitem{cirelli1984}
R~Cirelli, P~Lanzavecchia, and A~Mani\`{a}.
\newblock {Hamiltonian vector fields in Quantum Mechanics}.
\newblock {\em Nuovo Cimento B}, 79:271--283, 1984.

\bibitem{ClementeGallardo:2008p614}
J.~Clemente-Gallardo and G.~Marmo.
\newblock {Basics of Quantum Mechanics, Geometrization and Some Applications to
  Quantum Information}.
\newblock {\em International Journal of Geometric Methods in Modern Physics},
  05:989, 2008.

\bibitem{Ehrenfest1927}
P~Ehrenfest.
\newblock {Bemerkung \"{u}ber die angen\"{a}herte G\"{u}ltigkeit der
  klassischen Mechanik innerhalb der Quantenmechanik}.
\newblock {\em Zeitschrift f\"{u}r Physik}, 45(7-8):455--457, 1927.

\bibitem{Ercolessi2010}
E. Ercolessi, G. Marmo, and G. Morandi.
\newblock {From the Equations of Motion to the Canonical Commutation
  Relations}.
\newblock {\em Riv. Nuovo Cimento}, 033:401--590, May 2010.

\bibitem{Heslot1985a}
A.~Heslot.
\newblock {Quantum mechanics as a classical theory}.
\newblock {\em Physical Review D}, 31(6):1341--1348, 1985.

\bibitem{Kibble:1979p7279}
T~Kibble.
\newblock {Geometrization of quantum mechanics}.
\newblock {\em Communications in Mathematical Physics}, 65(2):189--201, 1979.

\bibitem{Koopman1931}
B.~O. Koopman.
\newblock {Hamiltonian systems and transformations in Hilbert space}.
\newblock {\em Proc. National Acad. Science}, 17:315--318, 1931.

\bibitem{Manko2002}
V~I Manko, G. Marmo, E~C~G. Sudarshan, and F. Zaccaria.
\newblock {Interference and entanglement: an intrinsic approach}.
\newblock {\em Journal of Physics A: Mathematical and General},
  35(33):7137--7157, August 2002.

\bibitem{marmo1996}
G~Marmo and G.~Vilasi.
\newblock {Symplectic structures and Quantum Mechanics}.
\newblock {\em Modern Physics Letters B}, 10(12):545--553, 1996.

\bibitem{Pancha1956}
S.~Pancharatnam.
\newblock {Generalized Theory of Interference, and Its Applications. Part I.
  Coherent Pencils}.
\newblock {\em Proc. Indian Acad. Sci. A}, 44:247--262, 1956.

\bibitem{Pancharatnam1956}
S~Pancharatnam.
\newblock {Generalized theory of interference and its applications. Part II.
  Partially coherent pencils}.
\newblock {\em Proceedings of the Indian Academy of Sciences, Section A},
  44(6):398--417, 1956.

\bibitem{Robertson1929}
H.~P. Robertson.
\newblock {The uncertainty principle}.
\newblock {\em Phys. Rev}, 34:163--164, 1929.

\bibitem{schrodinger1930heisenbergschen}
E~Schr\"{o}dinger.
\newblock {Zum Heisenbergschen Unsch\"{a}rfeprinzip}.
\newblock {\em Ber. Kgl. Akad. Wiss. Berlin}, 34:296--303, 1930.

\bibitem{Simon1983}
B.~Simon.
\newblock {Holonomy, the Quantum Adiabatic Theorem and Berry's phase}.
\newblock {\em Phys. Rev. Letters}, 51(24):2167--2170, 1983.

\bibitem{spera1993generalized}
M~Spera.
\newblock {On a generalized uncertainty principle, coherent states, and the
  moment map}.
\newblock {\em Journal of Geometry and Physics}, 12(3):165--182, 1993.

\bibitem{STROCCHI1966}
F.~Strocchi.
\newblock {Complex Coordinates and Quantum Mechanics}.
\newblock {\em Reviews of Modern Physics}, 38(1):36--40, January 1966.

\bibitem{Neumann1932}
J~von Neumann.
\newblock {Zur Operatorenmethode in der klassischen Mechanik}.
\newblock {\em Annals of Mathematics}, 33(3):587--642, 1932.

\end{thebibliography}

\end{document}